\documentclass[letterpaper, 10 pt, conference]{ieeeconf}  

\IEEEoverridecommandlockouts                              
\title{\LARGE \bf
Pursuit-evasion differential games of players with different speeds\\ in spaces of different dimensions}
\author{Shuai~Li, Chen~Wang and~Guangming~Xie
\thanks{This work was supported in part by grants from the National Natural Science Foundation of China (NSFC, No.61973007, 61633002). \emph{Corresponding author: C. Wang}}
\thanks{S. Li, C. Wang  and G. Xie are with the State Key Laboratory of Turbulence and Complex Systems, Intelligent Biomimetic Design Lab,
       College of Engineering, Peking University, Beijing 100871, China.
       {\tt\small \{shuaier, wangchen, xiegming\}@pku.edu.cn}}%
\thanks{C. Wang is also with the National Engineering Research Center of Software Engineering, Peking University, Beijing 100871, China.}%
}

\usepackage{amsmath}
\usepackage{amssymb}
\usepackage{xcolor}
\usepackage{graphicx}
\usepackage{bm}
\usepackage{cite}
\usepackage{cleveref}

\crefname{figure}{Fig.}{Fig.}

\newtheorem{theorem}{Theorem}

\newtheorem{lemma}{Lemma}
\newtheorem{remark}{Remark}
\newtheorem{problem}{Problem}

\newtheorem{assumption}{Assumption}

\begin{document}

\maketitle
\thispagestyle{empty}
\pagestyle{empty}

\begin{abstract}
We study pursuit-evasion differential games between a faster pursuer moving in 3D space and an evader moving in a plane. We first extend the well-known Apollonius circle to 3D space, by which we construct the isochron for the considered two players. Then both cases with and without a static target are considered and the corresponding optimal strategies are derived using the concept of isochron. In order to guarantee the optimality of the proposed strategies, the value functions are given and are further proved to be the solution of Hamilton-Jacobi-Isaacs equation. Simulations with comparison between the proposed strategies and other classical strategies are carried out and the results show the optimality of the proposed strategies. 
\end{abstract}

\section{INTRODUCTION}

Pursuit-evasion (PE) games have attracted a lot of interest in recent years, showing the wide applicability in modeling many adversarial problems, such as predator-prey problems in biology \cite{Peterson2021,Angelani2012}, and aerial combat in military \cite{Tang2021,Vlahov2018}. In classical one-pursuer-one-evader games \cite{Isaacs1965}, the pursuer tries to capture the evader as soon as possible while the evader thinks the contrary. Thus it can be modeled by a zero-sum game, which will be more complex when there exists other constraints. Originated from the seminal work by Isaacs \cite{Isaacs1965}, there are many kinds of variations of the problem, such as target-attacker-defender (TAD) problem \cite{Garcia2017,Zhang2021}, perimeter defending problem \cite{Shishika2018, Lee2020}, and multi-player problem \cite{Wang2020,Ramana2017}. 

Coming with various PE-like problems, many methods are developed to design strategies for players.
Taking the geometric properties of many PE or TAD problems, a plethora of methods based on isochron are proposed \cite{Weintraub2020,Yan2019,Zhou2016,Garcia2021,Yan2021}. Isochron is the set of points where a player can reach at the same time. And the intersection of isochrones of the pursuers and the evaders are a set of points they can arrive simultaneously. For players with identical speed, the intersection of isochrones is a perpendicular bisector between the pursuer and the evader \cite{Weintraub2020}, even in 3D space \cite{Yan2019}. When there are more than one pursuers, the intersection leading to Voronoi partitioning, which is used to design an effect capture strategy \cite{Zhou2016}. For players with different speeds, the intersection of isochrones of them usually forms a circle, called Apollonius circle \cite{Garcia2021,Yan2021}. With specific intersection of isochrones, different pursuit and evasion strategies can be obtained for different PE-like problems. However, these geometric approaches cannot guarantee the optimality of the strategy.

Differential games method by \cite{Isaacs1965} is born to solve this kind of game with continuous states. However, the difficulty in solving Hamilton-Jacobi-Isaacs (HJI) equation, especially when the dimension is high, limits the application of this method. Via modeling in differential games and calculating numerically, reach-avoid method is developed for TAD-like problems \cite{Huang2011,Chen2017,Landry2018}, where the player has to come to a region while avoiding another region. Besides, reinforcement learning (RL) is also used to solve these games or HJI equation in many works \cite{DeSouza2021,Wang2020}. However, these numerical methods have obvious weakness in computation time and solution accuracy.

Despite of considerable efforts in varying problem types and designing kinds of methods, most of existing works about PE problems, including most of the above works and others in the literature \cite{Manoharan2021,Zheng2021,Chen2016a,Ruiz2016,Yan2017}, assume that all the players move in a 2D plane.
This assumption of dynamics, although simplifying the analysis and solution of the problem, limits its application in real world. Many pursuit-evasion games take place in 3D space, such as intercepting a missile, chasing a flying bird. There is a little literature about PE problems in 3D space \cite{Yan2019, Garcia2020}, but they only consider the case that the players have the same constant speed. 

In this paper, we consider the PE problem with a faster pursuer moving in 3D space and a slower evader moving in a 2D plane. The crosses of different velocities and different motion spaces make the problem more complicated, and more practical as well. In nature, there are many flying predators who capture the terrestrial preys, like eagles capturing hares and rats. In military, unmanned aerial vehicles (UAVs) and missiles are often used to destroy ground targets like tanks. Our proposed model is suitable to formulate these real-world scenarios. Depending on whether the evader has a target to reach, we divide the problem into two cases, one is modeled as the classical PE game and another is modeled as TAD game. The latter can be used to study the problem when the hare has a hole, or the tank has a safe region.

In order to give an accurate solution, we make use of the concept of isochron and construct the intersection of isochrones for the pursuer moving in 3D space and the evader moving in a 2D plane, which is different from the classical Apollonius circle \cite{Garcia2021,Yan2021} in 2D plane. 
Then for PE problem without target, we find the farthest point on the Apollonius circle for the evader. Using this point we propose the optimal strategies of the players and the corresponding value of the game. For TAD problem when the target is available, counter intuitively, the optimal strategy for the evader is not to move to the target directly. Instead it should move to the closest point to the target on the circle. We further derive the optimal strategies for these two players and give the value of the game. Moreover, using differential games method, we show the verifications of the value functions and the optimal strategies of the two games, which are important to guarantee the optimality of the proposed strategies. Finally, the optimality of the proposed optimal strategies are illustrated by simulations.

The remainder of the paper is organized as follows. In \Cref{sec_formu}, we formulate the PE problem and give some useful preliminary results. 
Then, in \Cref{sec_TAD} and \Cref{sec_PE}, we address the two cases of the PE problem without and with a target, respectively; for each case, both the optimal strategy and value function are proposed and the corresponding verification is provided.
%
In \Cref{sec_simu}, we show up some simulations and the comparison with other strategies. Finally, \Cref{sec_conc} concludes the paper and raises future works.

\section{Game of players with different dimensions} \label{sec_formu}

\subsection{Problem Formulation}
Consider a pursuer $P$ moving in 3D space and an evader $E$ moving in a plane. The pursuer aims to capture the evader while the evader aims to escape from the pursuer. 
Without loss of generality, we establish Cartesian coordinates with $x,y$ axes in the evader's plane. Then, two players' positions are expressed as $P=(x_P,y_P,z_P)^\top\in \mathbb{R}^3$ and $E=(x_E,y_E)^\top\in \mathbb{R}^2$. Note that, in the following, we use $P$ and $E$ to represent the two players and their positions if there is no ambiguity. The whole state is $\bm{x}=(x_P,y_P,z_P,x_E,y_E)^\top \in \mathbb{R}^5$. We assume that the players can move freely in their own space, that is, their dynamics are given by
\begin{align}
	\begin{matrix}
		\dot{x}_P = u_x, & \dot{y}_P = u_y, & \dot{z}_P = u_z, \\
		\dot{x}_E = v_x, & \dot{y}_E = v_y, & 
	\end{matrix}
	\label{eq_dynamics}
\end{align}
where  $\bm{u}=(u_x,u_y,u_z)^\top,\bm{v}=(v_x,v_y)^\top$ are their control inputs, respectively. In this paper, we consider the case where the pursuer can move faster than the evader, which is along with most cases in nature. We make the following assumptions.
\begin{assumption} \label{assum_1}
    The speeds of both the pursuer and the evader are bounded, that is, there exist $u,v \in (0,+\infty)$ such that
    \begin{align}
        \sqrt{u_x^2+u_y^2+u_z^2}\leq u, ~~
        \sqrt{v_x^2+v_y^2}\leq v.
        \label{eq_speed_bound}
    \end{align}
\end{assumption}
\begin{assumption} \label{assum_2}
    The maximum speed of the pursuer is bigger than that of the evader, that is
    \begin{align}
        &\frac{u}{v} \triangleq \delta>1.
	    \label{eq_speed_ratio}
    \end{align}
\end{assumption}
~\par

Now we formulate the two pursuer-evasion differential games between players with different velocities and different motion spaces of interests.

\begin{problem}[PE game]\label{prob_1} 
    Consider two players, $P$ and $E$, modeled by (\ref{eq_dynamics}), where a pursuer moving in 3D space and an evader moving in a 2D plane. Under \Cref{assum_1} and \ref{assum_2}, find the optimal strategies for these two players to achieve (for $P$) or delay (for $E$) the capture.
\end{problem}

Besides, we also consider that there is a static target $T$.  When the evader moves to the target before being captured by the pursuer, the evader win the game. Otherwise, if the evader is captured by the pursuer before reaching the target, the pursuer win the game.

\begin{problem}[TAD game]\label{prob_2} 
    Consider two players, $P$ and $E$, modeled by (\ref{eq_dynamics}), where a pursuer moving in 3D space and an evader moving in a 2D plane, and a static target for the evader to approach. The objective of $E$ is to get close to the target without being captured, while the objective of $P$ is to capture $E$ at the farthest distance to the target as possible. Under \Cref{assum_1} and \ref{assum_2}, determine which player will win the game and find the corresponding optimal strategies.
\end{problem}

\subsection{Apollonius circle}
In the study of PE problem, isochron is an effective method to acquire the result of the game. In this paper, we derive the intersection of isochrones for 3D PE problem of players with different velocities. Unlike the works in \cite{Yan2019,Garcia2020} where the intersection of isochrones is a plane for players with the same velocities,  the intersection of isochrones for players in our problem is an extended form of Apollonius circle \cite{Isaacs1965}.
The area inside the circle is the dominant region of the evader since it can move to this region before the pursuer. While the area outside the circle is the dominant region of the pursuer.

Suppose $P$ and $E$ can simultaneously reach $(x,y,0)$ in the plane. 
If both of the players move at their maximum speed, then the speed ratio is $\delta$ and we have
\begin{align}
	&\frac{\sqrt{(x-x_P)^2+(y-y_P)^2+z_P^2}}{\sqrt{(x-x_E)^2+(y-y_E)^2}}=\delta, \\
	&\Big(x-\frac{\delta^2x_E-x_P}{\delta^2-1}\Big)^2 + \Big(y-\frac{\delta^2 y_E-y_P}{\delta^2-1}\Big)^2 \nonumber\\
	&~~~~= \frac{z_P^2}{\delta^2-1}+\frac{\delta^2[(x_P-x_E)^2+(y_P-y_E)^2]}{(\delta^2-1)^2}.
	\label{eq_A_circle}
\end{align}
Thus, the intersection of isochrones is a circle in the plane centered at
\begin{align}
	\bm{c} = \Big(\frac{\delta^2x_E-x_P}{\delta^2-1},\frac{\delta^2 y_E-y_P}{\delta^2-1},0\Big)^\top,
	\label{eq_center}
\end{align}
with radius
\begin{align}
	r = \sqrt{\frac{z_P^2}{\delta^2-1}+\frac{\delta^2[(x_P-x_E)^2+(y_P-y_E)^2]}{(\delta^2-1)^2}}.
	\label{eq_radius}
\end{align}
\begin{remark}
    When we suppose the evader can move in 3D space, the capture point can be in 3D space and similarly we will get a spherical isochron. We call this ``Apollonius sphere" (see Fig. \ref{fig_PE}). The circle \eqref{eq_A_circle} used in this paper is actually the intersection of the ``Apollonius sphere" and the plane, and we call it Apollonius circle as well.
\end{remark}

Noth that the intersection, Apollonius circle, is deternimined by the instantaneous positions of $P$ and $E$. So it is instantaneously calculated when implementing the optimal strategies in Section \ref{sec_PE} and \ref{sec_TAD}. Besides, the intersection of isochrones is obtained by assuming both players move at maximum speeds, which is in accordance with the optimal strategies (see Lemma \ref{lemma1} and \ref{lemma2}).


\subsection{Differential game}
Differential game is often used to describe the adversarial system. For the two pursuer-evasion games, the objectives of $P$ and $E$ are opposite, i.e. this is a zero-sum game. Define the cost function of this game as
\begin{align}
	J = \Phi(\bm{x}(t_f),\bm{u}(t_f),\bm{v}(t_f),t_f) + \int_0^{t_f} h(\bm{x},\bm{u},\bm{v},t) ~dt,
\end{align}
where $t_f$ is the terminal time, $\Phi(\cdot)$ is the terminal cost, and $h(\cdot)$ is the running cost.
The pursuer and the evader aim to find the optimal strategies to minimize or maximize the cost in the game, i.e.,
\begin{align}
	\min_{\bm{u}}\max_{\bm{v}} J(\bm{x},\bm{u},\bm{v},t).
\end{align}
%
The optimal cost is called the \textit{value} of the game and can be expressed as a function $V(\bm{x},t)=	\min_{\bm{u}}\max_{\bm{v}} J$. According to \cite{Isaacs1965}, $V$ should satisfy HJI equation
\begin{align}
	-\frac{\partial V}{\partial t} = \frac{\partial V}{\partial \bm{x}} (\bm{u}^*,\bm{v}^*)^\top + h(\bm{x},\bm{u}^*,\bm{v}^*,t),
	\label{eq_HJI}
\end{align}
where $\bm{u}^*,\bm{v}^*$ are the optimal strategies of the two players.

\section{PE game}\label{sec_PE}
We first consider the PE problem without target. 
In some real scenarios, target $T$ may be too far to be observed by the two players, or in some state that has no effect on the game. This corresponds to \Cref{prob_1}, the PE game, where $P$ aims to capture $E$ as soon as possible and $E$ wants to defer the capture. According to the objectives of $P$ and $E$ in Problem \ref{prob_1}, the cost function can then be 
\begin{align}
	J = \int_0^{t_f} dt
	\label{eq_J_t}
\end{align}
where $t_f$ is the time when
\begin{align}
	x_P=x_E, \; y_P=y_E, \; z_P=0.
\end{align}

Let $\mathcal{R}_0$ denote a set of position configurations
\begin{align}
	\mathcal{R}_0=\{\bm{x} \in \mathbb{R}^5  |(x_E-x_P)^2+(y_E-y_P)^2\neq 0\}.
\end{align}
When $\bm{x}\not\in \mathcal{R}_0$, $P$ is above $E$ and the Apollonius circle's center $\bm{c}$ is at $E$. Obviously, $E$ can move along with any direction and $P$ should follow $E$ while approaching the plane. It implies a fact that all strategies for $E$ are equivalent while $P$ should take an action following $E$, indicating that there is no need to decide an optimal strategies for this case.
Therefore, in the following, we consider the case when the initial position configuration is in the set  $\mathcal{R}_0$.

\begin{lemma} \label{lemma1}
	Consider differential game \eqref{eq_dynamics}\eqref{eq_speed_bound}\eqref{eq_speed_ratio}\eqref{eq_J_t} with $\bm{x}\in\mathcal{R}_0$, the optimal strategies of $P$ and $E$ are constant and their trajectories are straight lines.
\end{lemma}
\begin{proof}
	Let the Hamilton function be
	\begin{align}
		H = 1+\lambda_x u_x+\lambda_y u_y+\lambda_z u_z+\mu_x v_x+\mu_y v_y,
	\end{align}
	where $\bm\lambda = (\lambda_x, \lambda_y, \lambda_z, \mu_x, \mu_y)^\top$ is the co-state vector. Since $\dot{\bm\lambda}=\frac{\partial H}{\partial \bm{x}}=0$, the co-states are all constant. Then the optimal strategies are constant as well and the trajectories are straight lines.
\end{proof}
\begin{lemma} \label{lemma2}
    Consider differential game \eqref{eq_dynamics}\eqref{eq_speed_bound}\eqref{eq_speed_ratio}\eqref{eq_J_t} with $\bm{x}\in\mathcal{R}_0$, the optimal strategies of $P$ and $E$ are such that both of them use maximum speed to chase or escape.
\end{lemma}
\begin{proof}
    From Lemma \ref{lemma1}, the co-state vector are constant. For the pursuer, the optimal strategy can be obtained by 
    \begin{align}
        \bm{u}^*=\arg\min_{\bm{u}} H, ~~ s.t. \|\bm{u}\|\leq u.
    \end{align}
    Since $\bm{u}^*$ is constant w.r.t. time from Lemma \ref{lemma1}, this reduce to a linear programming where the objective is linear and the feasible region is convex. Thus the optimal solution is on the boundary of the feasible region, which means $\|\bm{u}\|= u$. Similar we can show that $\|\bm{v}\|=v$.
\end{proof}

From the above Lemmas, we know that $P$ and $E$ will move at maximum speed with constant direction when using optimal strategies. Then we use Apollonius circle \eqref{eq_A_circle} to determine the value function and the optimal strategy. Since $P$ and $E$ can move to the Apollonius circle simultaneously, the capture must happen on the circle under the two players' optimal strategies. For $E$, it should move to the farthest point $q^*$ on the circle to defer the capture, which is given by
\begin{align}
	q^* = (\frac{\delta^2x_E-x_P}{\delta^2-1}+r\cos\theta,\frac{\delta^2y_E-y_P}{\delta^2-1}+r\sin\theta,0)^\top,
\end{align}
where
\begin{align}
	&\cos\theta = \frac{x_E-x_P}{\sqrt{(x_E-x_P)^2+(y_E-y_P)^2}}, \\
	&\sin\theta = \frac{y_E-y_P}{\sqrt{(x_E-x_P)^2+(y_E-y_P)^2}}.
\end{align}
And the time for $P$ and $E$ moving to $q^*$ is given by
\begin{align}
	 t_E(\bm{x}) &= \frac{\|q^*-E\|}{v}= \frac{1}{v}(r+\|\bm{c}-E\|) \nonumber\\
	&= \frac{1}{v}\Big(r+\frac{\sqrt{(x_E-x_P)^2+(y_E-y_P)^2}}{\delta^2-1}\;\Big) \nonumber\\
	& = \frac{1}{v(\delta^2-1)}\Big(d+\sqrt{(\delta^2-1)z_P^2+\delta^2d^2}\;\Big),
	\label{eq_tE}
\end{align}
where $d=\sqrt{(x_E-x_P)^2+(y_E-y_P)^2}$.


In order to move to $q^*$ directly, the strategies of $P$ and $E$ are given by
\begin{align}
	&u_x^* = u\cos\phi\cos\theta, u_y^* = u\cos\phi\sin\theta, u_z^* = -u\sin\phi, \nonumber\\
	&v_x^* = v\cos\theta, v_y^* = v\sin\theta,
	\label{eq_optimal_u_PE}
\end{align}
where
\begin{align}
	\cos\phi &= \frac{\sqrt{(x_{q^*}-x_P)^2+(y_{q^*}-y_P)^2}}{\|P-q^*\|} \nonumber\\
	&= \frac{\sqrt{(x_{q^*}-x_P)^2+(y_{q^*}-y_P)^2}}{\delta\|q^*-E\|} \nonumber\\
	&= \frac{r+\|c-E\|+d}{\delta(r+\|c-E\|)} \nonumber\\
	&= \frac{\delta^2 d+\sqrt{(\delta^2-1)z_P^2+\delta^2 d^2}}{\delta(d+\sqrt{(\delta^2-1)z_P^2+\delta^2 d^2})}, \\
	\nonumber\\
	\sin\phi &= \frac{z_p}{||P-q^*||} = \frac{z_p}{\delta(r+\|c-E\|)} \nonumber\\
	&= \frac{(\delta^2-1)z_P}{\delta(d+\sqrt{(\delta^2-1)z_P^2+\delta^2 d^2})}.
\end{align}

\begin{theorem}[PE game] \label{th_1}
	For differential game \eqref{eq_dynamics}\eqref{eq_speed_bound}\eqref{eq_speed_ratio}\eqref{eq_J_t}, the optimal strategies are given by \eqref{eq_optimal_u_PE} and the corresponding value function is $V(\bm{x})=t_E(\bm{x})$, where $t_E(\bm{x})$ is given by \eqref{eq_tE}.
\end{theorem}

\begin{proof}
	We should show that the value function is $C^1$ and it satisfies HJI function \eqref{eq_HJI} with the proposed optimal strategies. 
	First, the gradient of the value function can be derived as
	\begin{align}
		\frac{\partial V}{\partial x_P} &= \frac{1}{(\delta^2-1)v}\Big(\frac{x_P-x_E}{d}+\frac{\delta^2 (x_P-x_E)}{\sqrt{(\delta^2-1)z_P^2+\delta^2 d^2}}\Big), \\
		\frac{\partial V}{\partial y_P}& = \frac{1}{(\delta^2-1)v}\Big(\frac{y_P-y_E}{d}+\frac{\delta^2 (y_P-y_E)}{\sqrt{(\delta^2-1)z_P^2+\delta^2 d^2}}\Big), \\
		\frac{\partial V}{\partial z_P} &= \frac{z_P}{v\sqrt{(\delta^2-1)z_P^2+\delta^2 d^2}}, \\
		\frac{\partial V}{\partial x_E}& = \frac{1}{(\delta^2-1)v}\Big(\frac{x_E-x_P}{d}+\frac{\delta^2 (x_E-x_P)}{\sqrt{(\delta^2-1)z_P^2+\delta^2 d^2}}\Big), \nonumber\\
		~~~~&=-\frac{\partial V}{\partial x_P}, \\
		\frac{\partial V}{\partial y_E} &= \frac{1}{(\delta^2-1)v}\Big(\frac{y_E-y_P}{d}+\frac{\delta^2 (y_E-y_P)}{\sqrt{(\delta^2-1)z_P^2+\delta^2 d^2}}\Big), \nonumber\\
		~~~~&=-\frac{\partial V}{\partial y_P}.
	\end{align}
	Since the optimal strategies are constant for $P$ and $E$, $\bm{x}\in\mathcal{R}_0$ holds until $t_f$. Thus $V$ is $C^1$.
	
	Then, substituting optimal strategies \eqref{eq_optimal_u_PE} into HJI equation \eqref{eq_HJI}, we get
	\begin{align}
		&\frac{\partial V}{\partial x_P}u_x^* +\frac{\partial V}{\partial y_P}u_y^* +\frac{\partial V}{\partial z_P}u_z^* + \frac{\partial V}{\partial x_E} v_x^* +\frac{\partial V}{\partial y_E} v_y^* +1 \nonumber\\
		= &\frac{\partial V}{\partial x_P}(u_x^*-v_x^*) + \frac{\partial V}{\partial y_P}(u_y^*-v_y^*) + \frac{\partial V}{\partial z_P}u_z^* +1 \nonumber\\
		= &\frac{\partial V}{\partial x_P}v\cos\theta(\delta\cos\phi-1)\nonumber\\
		& + \frac{\partial V}{\partial y_P}v\sin\theta(\delta\cos\phi-1) - \frac{\partial V}{\partial z_P}\delta v\sin\phi +1 \nonumber\\
		=&1+ \frac{\delta\cos\phi-1}{\delta^2-1}(\frac{1}{d}+\frac{\delta^2}{\sqrt{(\delta^2-1)z_P^2+\delta^2 d^2}})\nonumber\\
		&\big((x_P-x_E)\cos\theta+(y_P-y_E)\sin\theta\big) - \frac{\partial V}{\partial z_P}\delta v\sin\phi \nonumber\\
		=&1- \frac{d}{d+\sqrt{(\delta^2-1)z_P^2+\delta^2 d^2}}\nonumber\\
		&-\frac{\delta^2 d^2}{(d+\sqrt{(\delta^2-1)z_P^2+\delta^2 d^2})\sqrt{(\delta^2-1)z_P^2+\delta^2 d^2}}\nonumber\\
		& - \frac{(\delta^2-1)z_P^2}{(d+\sqrt{(\delta^2-1)z_P^2+\delta^2 d^2})\sqrt{(\delta^2-1)z_P^2+\delta^2 d^2}} \nonumber\\
		=&1-\frac{d}{d+\sqrt{(\delta^2-1)z_P^2+\delta^2 d^2}}
		-\frac{\sqrt{(\delta^2-1)z_P^2+\delta^2 d^2}}{d+\sqrt{(\delta^2-1)z_P^2+\delta^2 d^2}}\nonumber\\
		=&0.
	\end{align}
	Since $\frac{\partial V}{\partial t}=0$, the value function and the proposed optimal strategies satisfy HJI equation. Therefore, the value of the game is given by $V(\bm{x})=t_E(\bm{x})$ and the optimal strategies are \eqref{eq_optimal_u_PE}.
\end{proof}

\begin{remark}
    Note that the strategies defined by \eqref{eq_optimal_u_PE} are exactly feedback strategies and are optimal in the sense of Nash equilibrium. That is, if $P$ or $E$ change its strategies unilaterally, its cost will increase, which will be verified in Section \ref{sec_simu}. Thus, Problem \ref{prob_1} is solved.
\end{remark}

\section{TAD game} \label{sec_TAD}
In this section, we consider the case when there is a target, which corresponds to \Cref{prob_2}.
Without loss of generality, suppose the target is located at $O=(0,0,0)^\top$ and cannot move. $E$ aims to move to the target as close as possible in the final. The practical meaning is clear, $E$ should try its best to approach the target, expecting the lucky case when $P$ has some mistakes and cannot perform optimal strategy. For eagle-hare game, the target can be the hole of the hare. For UAV-tank game, the target may be some small area with electromagnetic interference, which is dangerous for UAV, or the target may be some object to be destroyed. As the game can be ended with two results, we first solve the game of kind \cite{Garcia2018}.

\subsection{Game of kind}
For the two results of the game, i.e.\ $P$ wins or $E$ wins, two terminal sets are defined
\begin{align}
	&S_P = \{\bm{x}\in \mathbb{R}^5|x_P=x_E,y_P=y_E,z_P=0\},\\
	&S_E = \{\bm{x}\in \mathbb{R}^5|x_E=0,y_E=0\}.
\end{align}
The initial configuration of two players' positions can be divided into two sets, $W_P$ and $W_S$, according to that who will win the game. In the first set $W_P$, $\bm{x}$ will first enter $S_P$, which means $P$ will win the game. While in the second set $W_E$, $\bm{x}$ will first enter $S_E$, which means $E$ will win the game. According to Apollonius circle, it is easy to determine these two sets.
\begin{align}
	&W_P = \{\bm{x}\in \mathbb{R}^5| \text{O is outside Apollonius circle}\},\\
	&W_E = \{\bm{x}\in \mathbb{R}^5| \text{O is inside Apollonius circle}\}, 
\end{align}
which are exactly
\begin{align}
	&W_P = \{\bm{x}\in \mathbb{R}^5| \delta^2(x_E^2+y_E^2)-(x_P^2+y_P^2+z_P^2)>0\},\\
	&W_E = \{\bm{x}\in \mathbb{R}^5| \delta^2(x_E^2+y_E^2)-(x_P^2+y_P^2+z_P^2)<0\}.
\end{align}
Then $B(\bm{x}) = \delta^2(x_E^2+y_E^2)-(x_P^2+y_P^2+z_P^2)$ is the barrier function. Thus we partially solve \Cref{prob_2} about determining which will win the game. Next we will give the optimal strategies of the players.

\subsection{Game of degree}
In the case of $\bm{x}\in W_E$, $E$ can easily move to the target before $P$ and win the game. Hoping to taking the chance to avoid the evader, the pursuer should move to the target as soon as possible. Thus the cost function for $P$ is
\begin{align}
	J_P = \sqrt{(x_P-x_E)^2 + (y_P-y_E)^2}\Big|_{t_f}.
\end{align}
And the optimal strategies of the two player are to move to the target directly. The value function is 
\begin{align}
	V(\bm{x})=(x_P^2+y_P^2+z_P^2)-\delta^2(x_E^2+y_E^2).
\end{align}
Due to space limits, we omit the simple verification of this case.

In the case of $\bm{x}\in W_P$, $P$ will catch $E$ and win the game. According to the objectives of $P$ and $E$ in Problem \ref{prob_2}, the cost function is
\begin{align}
	J = \Phi(t_f) = \|E-O\|\Big|_{t_f} = \sqrt{x_E^2+y_E^2}\Big|_{t_f},
	\label{eq_J_TAD}
\end{align}
and $P$ and $E$ will find strategies to maximize or minimize $J$,
which means $E$ will try its best to get close to $T$ while $P$ will prevent this. The optimal strategies and the value function in this case will be given and verified in the following.

When $\bm{x}\in W_P$, $E$ will finally be captured by $P$ and it should reduce the distance to the target at the end of the game. 
Since the running term (integral term) in \eqref{eq_J_TAD} is zero, similar to \Cref{lemma1} and \ref{lemma2}, the optimal strategies of $P$ and $E$ are constant and the maximum speeds are achieved as well.
Due to that the intersection of isochrones of the game is a circle, the optimal strategy of $E$ is not to move to the target directly. In order to get close to the target in the final, $E$ should move to the nearest point to the target on the circle. Denote this point as $g^*$,
\begin{align}
	g^* = \Big(\frac{\delta^2x_E-x_P}{\delta^2-1}-r\cos\psi,\frac{\delta^2 y_E-y_P}{\delta^2-1}-r\sin\psi,0\Big)^\top,
\end{align}
where
\begin{align}
	\cos\psi = \frac{\delta^2x_E-x_P}{\sqrt{(\delta^2x_E-x_P)^2+(\delta^2 y_E-y_P)^2}}, \\
	\sin\psi = \frac{\delta^2 y_E-y_P}{\sqrt{(\delta^2x_E-x_P)^2+(\delta^2 y_E-y_P)^2}}.
\end{align}
Then, if $P$ and $E$ both move to $g^*$, they will meet there. We can derive the final distance from $E$ to $T$ as
\begin{align}
	L_f(\bm{x}) = \|g^*\| = \|\bm{c}\|-r 
	= \frac{d_\delta}{\delta^2-1}-r,
	\label{eq_Lf}
\end{align}
where $r$ is given in \eqref{eq_radius} and 
\begin{align}
	&d_\delta = \sqrt{(\delta^2x_E-x_P)^2+(\delta^2 y_E-y_P)^2}.
\end{align}
Since $\bm{x}\in W_P$, $d_\delta\neq 0$.

When P and E move to $g^*$ directly, the strategies are derived in the following.

For E, 
\begin{align}
	g^*-E = \big(\frac{x_E-x_P}{\delta^2-1}-r\cos\psi,\frac{y_E-y_P}{\delta^2-1}-r\sin\psi,0\big)^\top,
\end{align}
and its strategy to move to $g^*$ is 
\begin{align}
	v_x^* = v\cos\alpha,\; v_y^* = v\sin\alpha,
	\label{eq_v_TAD}
\end{align}
where
\begin{align}
	\cos\alpha &= \frac{\frac{x_E-x_P}{\delta^2-1}-r\cos\psi}{\sqrt{(\frac{x_E-x_P}{\delta^2-1}-r\cos\psi)^2+(\frac{y_E-y_P}{\delta^2-1}-r\sin\psi)^2}} \nonumber\\
	&= \frac{x_E-x_P-r(\delta^2-1)\cos\psi}{\sqrt{d^2+(\delta^2-1)^2r^2-2r(\delta^2-1)f}}, \\
	\sin\alpha & = \frac{\frac{y_E-y_P}{\delta^2-1}-r\sin\psi}{\sqrt{(\frac{x_E-x_P}{\delta^2-1}-r\cos\psi)^2+(\frac{y_E-y_P}{\delta^2-1}-r\sin\psi)^2}} \nonumber\\
	&= \frac{y_E-y_P-r(\delta^2-1)\sin\psi}{\sqrt{d^2+(\delta^2-1)^2r^2-2r(\delta^2-1)f}}, 
\end{align}
and
\begin{align}
	f=(x_E-x_P)\cos\psi+(y_E-y_P)\sin\psi.
\end{align}

For P,
\begin{align}
	g^*-P = &\Big(\frac{\delta^2(x_E-x_P)}{\delta^2-1}-r\cos\psi,\nonumber\\
	&~~~~\frac{\delta^2(y_E-y_P)}{\delta^2-1}-r\sin\psi,-z_P\Big)^\top,
\end{align}
and its strategy to move to $g^*$ is
\begin{align}
	u_x^*=u\cos\beta\cos\gamma, u_y^*=u\sin\beta\cos\gamma, u_z^*=-u\sin\gamma,
	\label{eq_u_TAD}
\end{align}
where
\begin{align}
	&\sin\gamma = \frac{z_P}{\|g^*-P\|} \nonumber\\
	= &\frac{(\delta^2-1)z_P}{\sqrt{\delta^4d^2+(\delta^2-1)^2r^2-2r\delta^2(\delta^2-1)f+(\delta^2-1)^2 z_P^2}},\\
	&\cos\gamma = \sqrt{1-\sin^2\gamma} \nonumber\\
	= &\frac{\sqrt{\delta^4d^2+(\delta^2-1)^2r^2-2r\delta^2(\delta^2-1)f}}{\sqrt{\delta^4d^2+(\delta^2-1)^2r^2-2r\delta^2(\delta^2-1)f+(\delta^2-1)^2 z_P^2}},
\end{align}
\begin{align}
	&\cos\beta = \frac{\frac{\delta^2(x_E-x_P)}{\delta^2-1}-r\cos\psi}{\sqrt{(\frac{\delta^2(x_E-x_P)}{\delta^2-1}-r\cos\psi)^2+(\frac{\delta^2(y_E-y_P)}{\delta^2-1}-r\sin\psi)^2}} \nonumber\\
	&= \frac{\delta^2(x_E-x_P)-r(\delta^2-1)\cos\psi}{\sqrt{\delta^4d^2+(\delta^2-1)^2r^2-2r\delta^2(\delta^2-1)f}}, \\
	&\sin\beta = \frac{\frac{\delta^2(y_E-y_P)}{\delta^2-1}-r\sin\psi}{\sqrt{(\frac{\delta^2(x_E-x_P)}{\delta^2-1}-r\cos\psi)^2+(\frac{\delta^2(y_E-y_P)}{\delta^2-1}-r\sin\psi)^2}} \nonumber\\
	&= \frac{\delta^2(y_E-y_P)-r(\delta^2-1)\sin\psi}{\sqrt{\delta^4d^2+(\delta^2-1)^2r^2-2r\delta^2(\delta^2-1)f}}.
\end{align}

\begin{remark}
    The target point $g^*$ and the strategies moving to $g^*$ are dependent on the current positions of $P$ and $E$. Thus, the strategies are feedback strategies and optimal in the sense of Nash equilibrium by the following Theorem.
\end{remark}

\begin{theorem}[TAD game]
	Consider differential game \eqref{eq_dynamics}\eqref{eq_speed_bound}\eqref{eq_speed_ratio}\eqref{eq_J_TAD} with $\bm{x}\in W_P$, the value function of the game is $V(\bm{x})=L_f(\bm{x})$ and the optimal strategies are given by \eqref{eq_u_TAD}\eqref{eq_v_TAD}.
\end{theorem}

\begin{proof}
	Similar to the proof of \Cref{th_1}, we should show that value function is $C^1$ and satisfies HJI function with optimal strategies.
	
	First, the gradient of $V$ can be calculated as 
	\begin{align}
		&\frac{\partial V}{\partial x_P}=\frac{x_P-\delta^2 x_E}{(\delta^2-1)d_\delta}-\frac{\delta^2(x_P-x_E)}{(\delta^2-1)\sqrt{(\delta^2-1)z_P^2+\delta^2 d^2}}, \\
		&\frac{\partial V}{\partial y_P}=\frac{y_P-\delta^2 y_E}{(\delta^2-1)d_\delta}-\frac{\delta^2(y_P-y_E)}{(\delta^2-1)\sqrt{(\delta^2-1)z_P^2+\delta^2 d^2}}, \\
		&\frac{\partial V}{\partial z_P}=-\frac{z_P}{\sqrt{(\delta^2-1)z_P^2+\delta^2 d^2}},\\
		&\frac{\partial V}{\partial x_E}=\frac{\delta^2(\delta^2 x_E-x_P)}{(\delta^2-1)d_\delta}-\frac{\delta^2(x_E-x_P)}{(\delta^2-1)\sqrt{(\delta^2-1)z_P^2+\delta^2 d^2}}, \\
		&\frac{\partial V}{\partial y_E}=\frac{\delta^2(\delta^2 y_E-y_P)}{(\delta^2-1)d_\delta}-\frac{\delta^2(y_E-y_P)}{(\delta^2-1)\sqrt{(\delta^2-1)z_P^2+\delta^2 d^2}}.
	\end{align}
	With $\bm{x}\in W_P$, one can obtain that the denominators are not zero until $E$ is captured by $P$ at $g^*$.
	
	Then, according to \eqref{eq_HJI}, since $\frac{\partial V}{\partial t}=0$, the other side must be zero. Let
	\begin{align}
		H=\frac{\partial V}{\partial x_P}u_x^* +\frac{\partial V}{\partial y_P}u_y^* +\frac{\partial V}{\partial z_P}u_z^* + \frac{\partial V}{\partial x_E} v_x^* +\frac{\partial V}{\partial y_E} v_y^*.
	\end{align}
	Using some symmetrical properties, we calculate $H$ by dividing it into two parts, $\frac{\partial V}{\partial x_P}u_x^* +\frac{\partial V}{\partial y_P}u_y^* +\frac{\partial V}{\partial z_P}u_z^*$  and $\frac{\partial V}{\partial x_E} v_x^* +\frac{\partial V}{\partial y_E} v_y^*$.
	
	We first deal with first part about $P$. 
	\begin{align}
		&\frac{\partial V}{\partial x_P}u_x^*+\frac{\partial V}{\partial y_P}u_y^*\nonumber\\ = & \frac{x_P-\delta^2 x_E}{(\delta^2-1)d_\delta}u_x^*+\frac{y_P-\delta^2 y_E}{(\delta^2-1)d_\delta}u_y^* \nonumber\\
		&-\frac{\delta^2(x_P-x_E)}{(\delta^2-1)\sqrt{(\delta^2-1)z_P^2+\delta^2 d^2}}u_x^* \nonumber\\
		&-\frac{\delta^2(y_P-y_E)}{(\delta^2-1)\sqrt{(\delta^2-1)z_P^2+\delta^2 d^2}} u_y^* \nonumber\\
		=& \frac{-\delta^2f+r(\delta^2-1)}{(\delta^2-1)^2 d_{Pg}}u +\frac{\delta^4d^2-r\delta^2(\delta^2-1)f}{r(\delta^2-1)^3 d_{Pg}}u,
	\end{align}
	where 
	\begin{align}
		&d_{Pg} = \|g^*-P\| \nonumber\\
		=&\frac{\sqrt{\delta^4d^2+(\delta^2-1)^2r^2-2r\delta^2(\delta^2-1)f+(\delta^2-1)^2 z_P^2}}{\delta^2-1}.
	\end{align}
	And
	\begin{align}
		&\frac{\partial V}{\partial z_P}u_z^* = \frac{z_P^2}{r(\delta^2-1) d_{Pg}}u.
	\end{align}
	Then one can have,
	\begin{align}
		\frac{\partial V}{\partial x_P}u_x^*+\frac{\partial V}{\partial y_P}u_y^*+\frac{\partial V}{\partial z_P}u_z^*=\frac{d_{Pg}}{r(\delta^2-1)}\delta v.
	\end{align}
	
	Second, we turn to the second part about $E$.
	\begin{align}
		&\frac{\partial V}{\partial x_E}v_x^*+\frac{\partial V}{\partial y_E}v_y^* \nonumber\\
		=& \frac{\delta^2(\delta^2 x_E-x_P)}{(\delta^2-1)d_\delta}v_x^* + \frac{\delta^2(\delta^2 y_E-y_P)}{(\delta^2-1)d_\delta}v_y^* \nonumber\\
		&-\frac{\delta^2(x_E-x_P)}{(\delta^2-1)\sqrt{(\delta^2-1)z_P^2+\delta^2 d^2}}v_x^* \nonumber\\
		&-\frac{\delta^2(y_E-y_P)}{(\delta^2-1)\sqrt{(\delta^2-1)z_P^2+\delta^2 d^2}}v_y^* \nonumber\\
		=& \frac{\delta^2 f-r\delta^2(\delta^2-1)}{(\delta^2-1)^2 d_{Eg}}v + \frac{-\delta^2d^2+r\delta^2(\delta^2-1)f}{r(\delta^2-1)^3 d_{Eg}}v \nonumber\\
		&= \frac{-\delta^2 d_{Eg}}{r(\delta^2-1)}v,
	\end{align}
	where
	\begin{align}
		d_{Eg} &= \|g^*-E\| \nonumber\\
		&=\frac{\sqrt{d^2+(\delta^2-1)^2r^2-2r(\delta^2-1)f}}{\delta^2-1}.
	\end{align}
	
	Thus,
	\begin{align}
		H = \frac{d_{Pg}}{r(\delta^2-1)}\delta v+\frac{-\delta^2 d_{Eg}}{r(\delta^2-1)}v = \frac{\delta v(d_{Pg}-\delta d_{Eg})}{r(\delta^2-1)}.
	\end{align}
	Since $g^*$ is on the Apollonius circle, $\frac{d_{Pg}}{d_{Eg}} =\delta $, which means $H=0$. We have shown that $V(\bm{x})=L_f(\bm{x})$ is $C^1$ and is the solution of HJI equation, for which \eqref{eq_u_TAD}\eqref{eq_v_TAD} are optimal strategies of the differential game.
\end{proof}

Thus, we solve the remainder of \Cref{prob_2} by giving optimal strategies of two players for both cases.

\begin{figure}[tpb]
	\centering
	\includegraphics[width=1\linewidth]{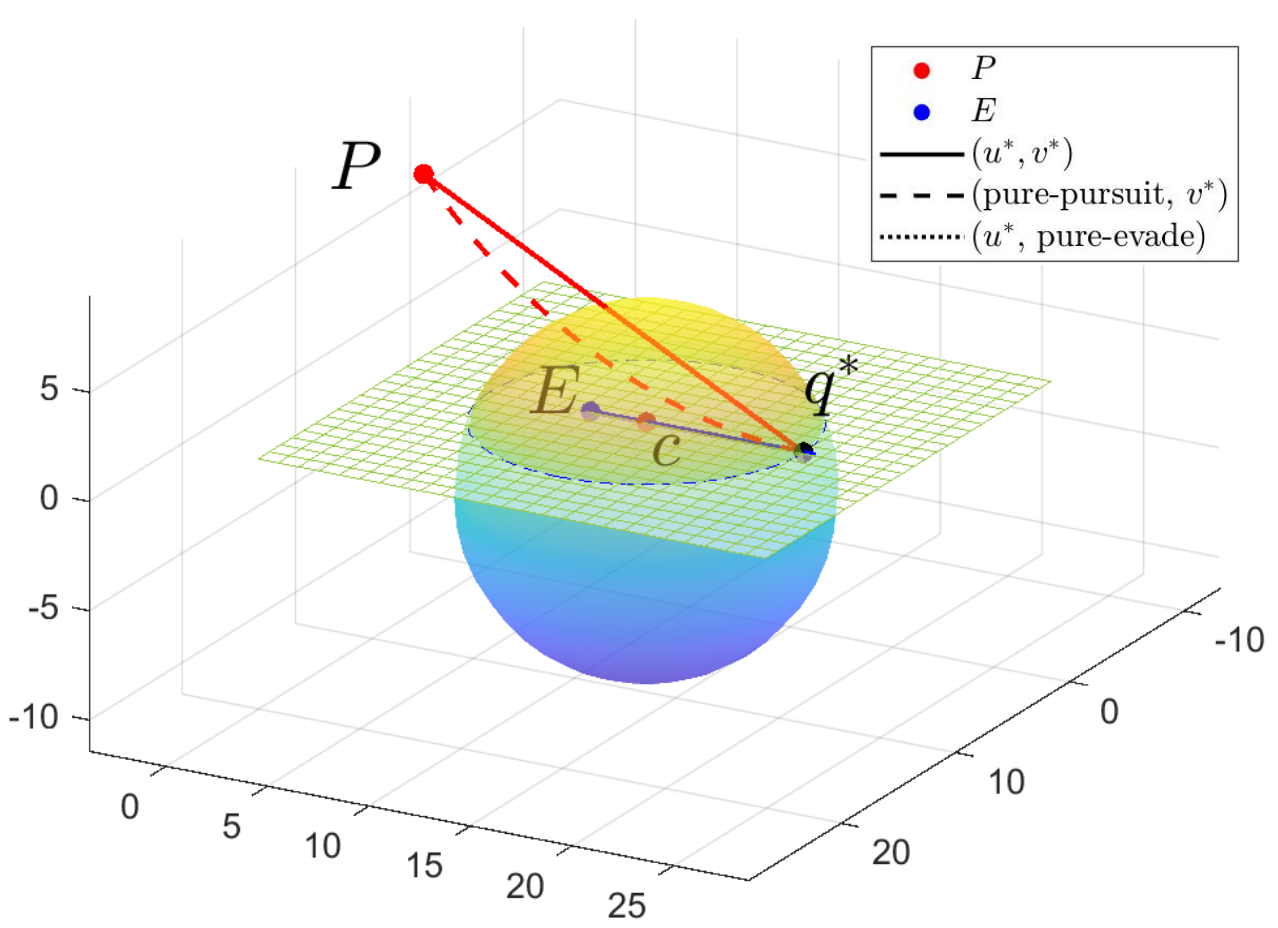}
	\caption{Results of PE game under different strategies. The initial positions of $P$ and $E$ are marked. Red lines and blue lines represent the trajectories of $P$ and $E$ respectively. Besides, the initial Apollonius circle are computed using the initial positions of $P$ and $E$, for which the center $c$, and the goal point $q^*$ are marked. 
	The ``Apollonius sphere'' are also plotted and the corresponding Apollonius circle in the plane are marked as dash curves.}
	\label{fig_PE}
\end{figure}

\section{SIMULATION RESULTS} \label{sec_simu}
In this section, we will give some simulations to show that our proposed strategies are optimal. Both two games are simulated under the optimal strategies with a comparison to other strategies, pure-pursuit, pure-evade and task-oriented. Pure-pursuit strategy means that the pursuer moves to the evader directly while pure-evade strategy means that the evader just moves to the opposite direction of the pursuer in the plane. Task-oriented strategy means that the player moves to the target directly.

First, consider the PE game between the pursuer and the evader, i.e., there is no target or the target is too far to be observed. Let the initial positions be $P=(8.00, 1.42, 9.41)^\top, E=(7.92, 9.60,0)^\top$ and the speed ratio is $\delta=2$. In this example, $P$ should capture $E$ as soon as possible (minimize $J$ in \eqref{eq_J_t}) while $E$ aims to delay the capture (maximize $J$ in \eqref{eq_J_t}). 
In order to illustrate the optimality of the proposed strategy \eqref{eq_optimal_u_PE}, pure-pursuit and pure-evade strategies are simulated as well. The trajectories and the costs are shown in \Cref{fig_PE} and \Cref{Tab_1}, from which we can see that the proposed strategies are optimal in the sense of Nash equilibrium, and the pure-evade strategy for $E$ is actually equivalent to $\bm{v}^*$. Besides, when both players use \eqref{eq_optimal_u_PE}, the cost is consistent with \eqref{eq_tE}.

Second, consider that $T$ can be observed by players and $E$ should move to $T$ as possible as it can. The initial positions are $P=(6.97, 5.83, 9.15)^\top, E=(9.80, 0.33, 0)^\top, T=(0,0,0)^\top$, which is the case $\bm{x}\in W_P$. And the speed ratio is $\delta=2$. 
Here $P$ should maximize $J$ in \eqref{eq_J_TAD} while $E$ should minimize it.
Using \eqref{eq_u_TAD}\eqref{eq_v_TAD}, we can get the optimal strategies $(\bm{u}^*,\bm{v}^*)$. Another two strategies, pure-pursuit/pure-evasion and task-oriented are simulated to compare the performance. As shown in \Cref{fig_TAD123}, \Cref{fig_TAD345}, and \Cref{Tab_2}, the proposed optimal strategy has the best performance. Neither $P$ or $E$ can improve its performance (i.e., increase or decrease the final distance to the target respectively) by changing its strategy unilaterally. Besides, when $P$ uses \eqref{eq_u_TAD} and $E$ uses \eqref{eq_v_TAD}, the cost is consistent with \eqref{eq_Lf}.

\begin{table}[th]
	\centering
	\begin{tabular}{c c c}
		\hline
		P's strategy & E's strategy & Time cost \\
		\hline
		pure-pursuit & $\bm{v}^*$ & 11.05 \\
		$\bm{u}^*$ & pure-evade & 10.43 \\
		$\bm{u}^*$ & $\bm{v}^*$ & 10.43 \\
		\hline
	\end{tabular}
	\caption{Simulation results of PE game under different strategies. $P$ aims to minimize time cost while $E$ aims to maximize it.}
	\label{Tab_1}
\end{table}

\begin{table}[th]
	\centering
	\begin{tabular}{c c c}
		\hline
		P's strategy & E's strategy & Final distance to the target \\
		\hline
		$\bm{u}^*$ & task-oriented & 4.2341 \\
		$\bm{u}^*$ & pure-evade & 19.4309 \\
		$\bm{u}^*$ & $\bm{v}^*$ & \textbf{4.1481}\\
		pure-pursuit & $\bm{v}^*$ & 3.2272 \\
		task-oriented & $\bm{v}^*$ & 0.0046 \\
		\hline
	\end{tabular}
	\caption{Simulation results of TAD game under different strategies. $P$ aims to maximize the final distance to the target while $E$ aims to minimize it.}
	\label{Tab_2}
\end{table}

\begin{figure}[tpb]
	\centering
	\includegraphics[width=1\linewidth]{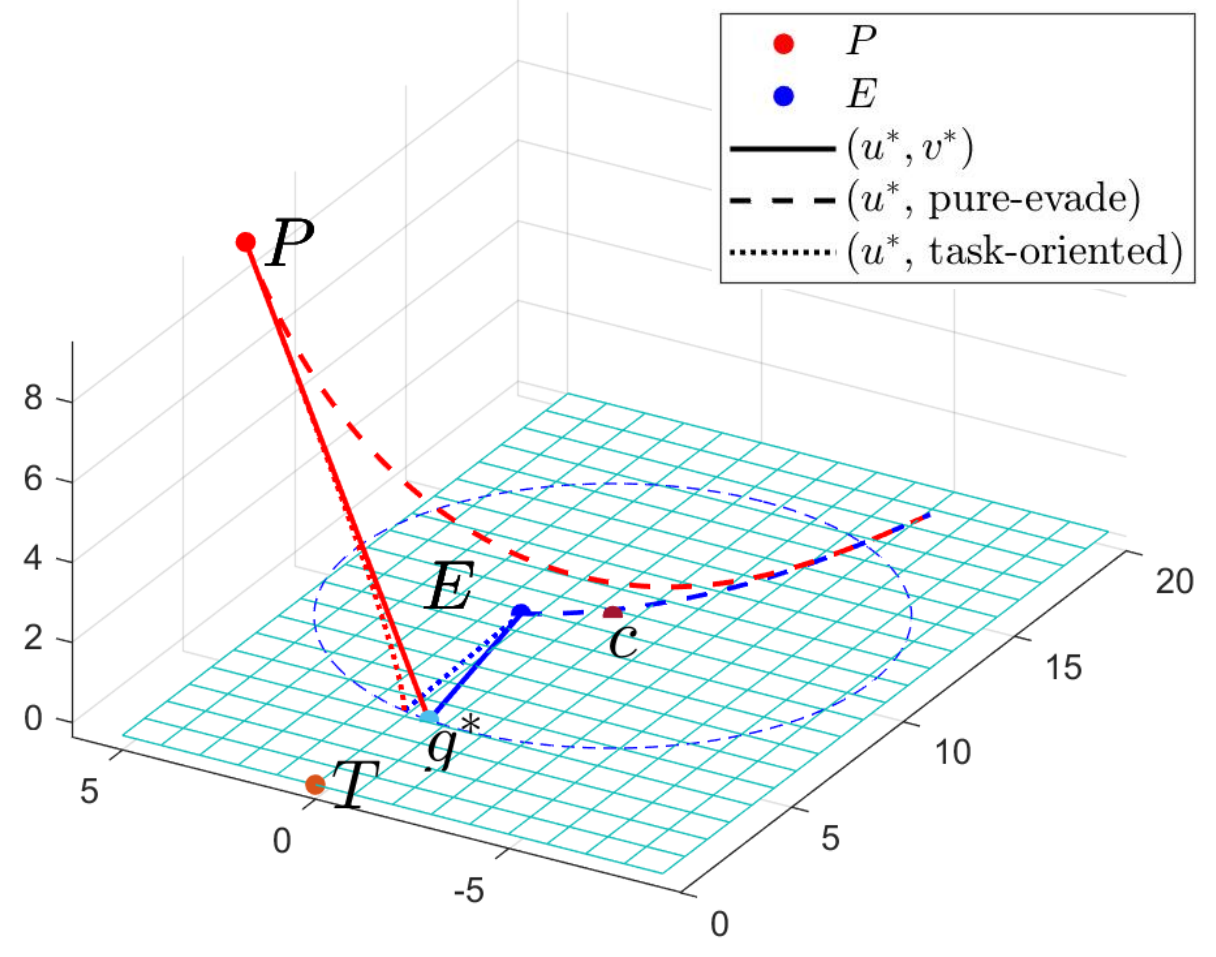}
	\caption{Results of TAD game under different strategies (Part 1). The initial positions of $P$ and $E$ are marked. Red lines and blue lines represent the trajectories of $P$ and $E$ respectively. Besides, the initial Apollonius circle are computed using the initial positions of $P$ and $E$, for which the center $c$, and the goal point $g^*$ are marked.  
	}
	\label{fig_TAD123}\vspace{-0.2cm}
\end{figure}

\section{CONCLUSIONS} \label{sec_conc}
In this paper, we investigate the pursuit-evasion problem between a faster pursuer moving in 3D space and a slower evader moving in a 2D plane. According to whether there is a static target, the game is modeled as PE game or TAD game, respectively. Both cases are considered and the corresponding optimal strategies are given using the idea of isochron. In order to demonstrate the optimality of the proposed strategies, we derive the value functions and further show that Hamilton-Jacobi-Isaacs equations are satisfied by the proposed strategies and value functions. Simulations are carried out, where some other strategies are compared to the proposed ones, and finally verify the optimality of the proposed strategies.
Future work will extend the results to the game of multiple players in 3D and 2D spaces. Besides, the problem containing an active target is also interesting.

\begin{figure}[tpb]
	\centering
	\includegraphics[width=1\linewidth]{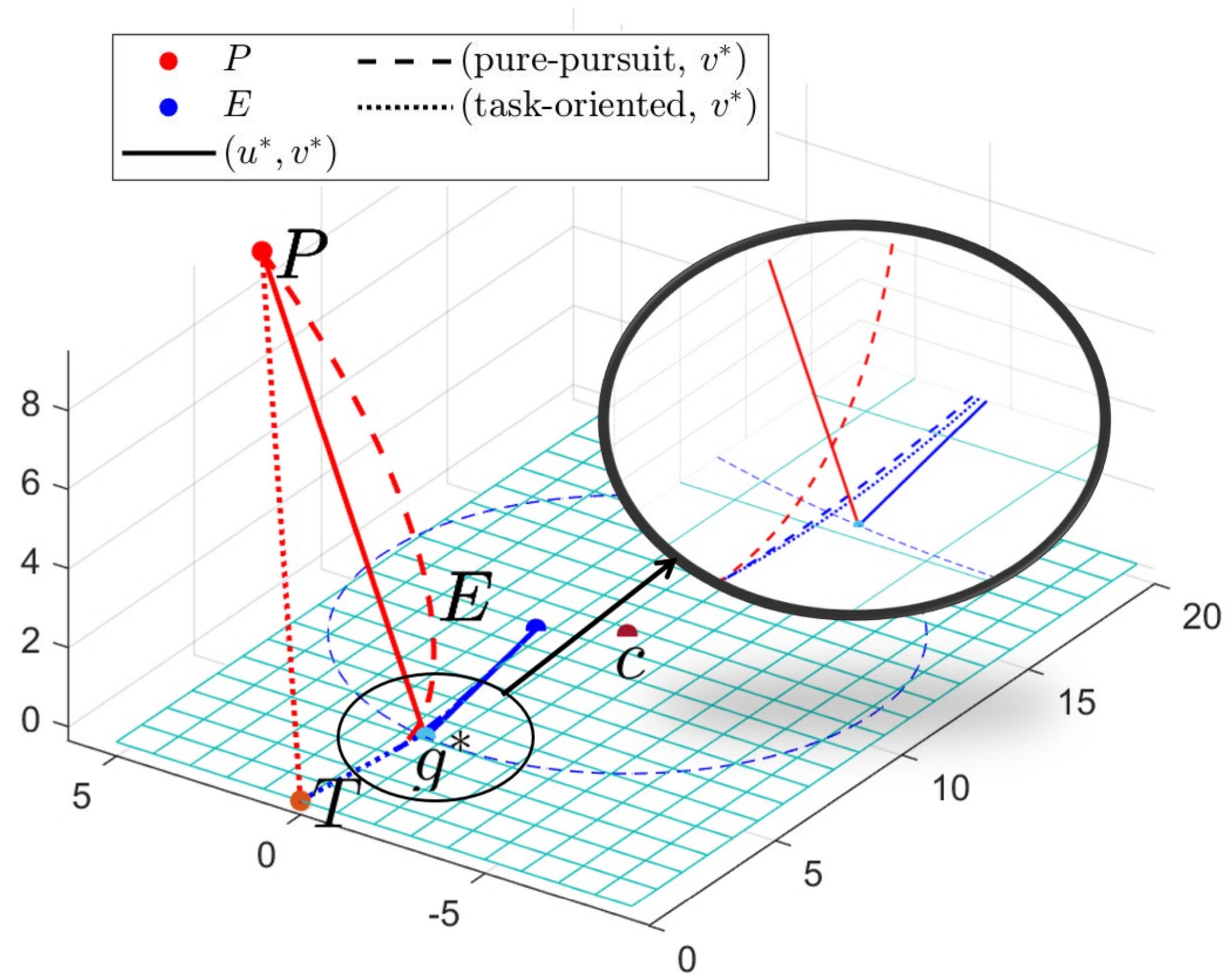}
	\caption{Results of TAD game under different strategies. The initial positions of $P$ and $E$ are marked. Red lines and blue lines represent the trajectories of $P$ and $E$ respectively. Besides, the initial Apollonius circle are computed using the initial positions of $P$ and $E$, for which the center $c$, and the goal point $g^*$ are marked.  
	}
	\label{fig_TAD345}
\end{figure}

\addtolength{\textheight}{-3cm}
%
%
%
%

%
%
%
%

\bibliographystyle{IEEEtran}        
\bibliography{shuai}

\end{document}